\documentclass[11pt]{article}

\usepackage{amsfonts}
\usepackage{amsmath,amsthm}
\usepackage{graphicx}

\usepackage{mathrsfs}
\usepackage{mathtools}

\usepackage{amssymb}
\usepackage{latexsym}

\newtheorem{theorem}{Theorem}[section]

\newtheorem{lemma}[theorem]{Lemma}
\newtheorem{proposition}[theorem]{Proposition}

\DeclareMathAlphabet{\mathpzc}{OT1}{pzc}{m}{it}
\DeclareMathAlphabet\EuFrak{U}{euf}{m}{n}	
\SetMathAlphabet\EuFrak{bold}{U}{euf}{b}{n}	

\newcommand{\bsx}{{\boldsymbol x}}
\newcommand{\bsy}{{\boldsymbol y}}

\newcommand{\bsp}{{\boldsymbol p}}

\def\CC{\mathbb{C}}

\def\RR{\mathbb{R}}

\newcommand{\ie}{{\it{i.e.\ }}}
\newcommand{\eg}{{\it{e.g.\ }}}

\newcommand{\mC}{\mathcal C}
\newcommand{\mD}{\mathcal D}

\newcommand{\mL}{\mathcal L}

\newcommand{\mO}{\mathcal O}

\newcommand{\efV}{\EuFrak{V}}
\newcommand{\efW}{\EuFrak{W}}

\newcommand{\cC}{{\cal C}}
\newcommand{\cD}{{\cal D}}

\newcommand{\bex}{\boldsymbol{e}}
\newcommand{\bfy}{\boldsymbol{f}}
\newcommand{\bg}{\boldsymbol{g}}

\newcommand{\bl}{\boldsymbol{l}}
\newcommand{\bm}{\boldsymbol{m}}
\newcommand{\bn}{\boldsymbol{n}}
\newcommand{\bp}{\boldsymbol{p}}

\newcommand{\bx}{\boldsymbol{x}}
\newcommand{\by}{\boldsymbol{y}}

\newcommand{\bA}{\boldsymbol{A}}

\newcommand{\bnabla}{\boldsymbol{\nabla}}

\newcommand{\brho}{{\boldsymbol{\rho}}}

\theoremstyle{definition}

\theoremstyle{remark}

\numberwithin{equation}{section}

\newcommand{\be}{\begin{equation}}
\newcommand{\ee}{\end{equation}}

\begin{document}

\title{Charges in light cones and quenched infrared radiation}

\author{Detlev Buchholz${}^a$, Fabio Ciolli${}^b$, Giuseppe Ruzzi${}^c$ and   
Ezio Vasselli${}^c$ \\[10pt]
\small  
${}^a$ Mathematisches Institut, Universit\"at G\"ottingen, \\
\small Bunsenstr.\ 3-5, 37073 G\"ottingen, Germany\\[5pt]
\small
${}^b$ Dipartimento di Matematica e Informatica, Universit{\`a}
della Calabria \\
\small Via Pietro Bucci - Cubo 30B, 87036 Rende (CS), Italy \\[5pt] 
\small
${}^c$ 
Dipartimento di Matematica, Universit{\`a} di Roma Tor Vergata \\
\small Via della Ricerca Scientifica 1, 00133 Roma, Italy \\
}
\date{}

\maketitle

\vspace*{-5mm}
\begin{abstract} \noindent
  The creation of electrically charged states and the resulting
  electromagnetic fields are considered in space-time regions in
  which such experiments can actually be carried out, namely in
  future-directed light cones. Under the simplifying assumption of
  external charges, charged states are formed from neutral pairs of
  opposite charges, with one charge being shifted to light-like infinity.
  It thereby escapes observation.
  Despite the fact that this charge moves asymptotically
  at the speed of light, the resulting electromagnetic field has a
  well-defined energy operator that is bounded from below. Moreover, 
  due to the spatiotemporal restrictions, the transverse
  electromagnetic field (the radiation) has no infrared
  singularities in the light cone. They are quenched and the observed 
  radiation can be described by states in the Fock space of photons.
  The longitudinal field between the charges (giving rise to
  Gauss's law) disappears for inertial observers in an instant.
  This is consistent with the fact that there is no evidence for
  the existence of longitudinal photons.
  The results show that the restrictions of operations and
  observations to light cones, which are dictated by the arrow of time,
  amount to a Lorentz-invariant infrared cutoff.
\end{abstract}

\section{Introduction}
\label{sec1}
\setcounter{equation}{0}

We consider in this article the construction
of electrically charged states in light cones  
from neutral pairs of opposite charges. This is accomplished 
by transporting one of the
charges to light-like infinity. It thereby 
disappears in the space-like complement
of any observer in the cone, leaving behind the
other charge, whence a charged state. Such a construction of charged states  
was proposed in~\cite{BuRo2014} in order to avoid  
the well known infrared problems
appearing in Minkowski space theories. It takes into account the fact 
that observations are only possible in future directed light cones,
they are impossible in the past. More precisely, they cannot be
made up for by returning to the~past. This basic 
fact implies that a resolution of the infrared
properties of radiation, predicted by the theory,
is beyond the limits of any experiment.

\medskip
To simplify the analysis, we restrict our attention to external charges
that act on the electromagnetic field but are not influenced by it. In
this framework, we determine the energetic properties of the field
generated by the transport of a charge in a given pair to light-like
infinity. Even though this charge moves asymptotically at the speed of light,
the resulting field induces representations of the observables in which the
spacetime translations are unitarily implemented and fulfill the relativistic
spectrum condition. Moreover, the subalgebras of the observables in light
cones are normal with respect to the vacuum (Fock) representation, \ie
the clouds of photons created by this procedure
do not appear as superselected.
These observations complement and enhance related 
results of Cadamuro and Dybalski \cite{CaDy}. Furthermore, 
the longitudinal photons that form the gauge bridges between
charges disappear in an instant and are virtual in this sense.
Our results confirm the ideas on which the
discussion in \cite{BuRo2014} is based. They show that the restriction of
operations and observations to light cones acts like a Lorentz-invariant
infrared cutoff. 

\medskip 
We adopt in the present analysis the framework 
in~\cite{BuCiRuVa2021,BuCiRuVa2022}, where infrared singular
charged states
were constructed in Minkowski space. These 
were obtained by shifting in neutral pairs of opposite charges
one of the charges to space-like infinity. 
This approach is based on a
concrete version of the universal C*-algebra of the electromagnetic
field, presented in~\cite{BuCiRuVa2016}.
The external charges are described by 
outer automorphisms, which are implemented by unitary operators
in an extended Gupta-Bleuler framework.
For the validity of Gauss's law, the
gauge fields in that setting are decisive.
They provide information about the value of electrical charges
based on flux measurements of the electromagnetic field. 
Whereas the impact of the charges on the
electromagnetic field can be studied in
that framework, it does not allow to describe
the effects of changes of their location. This location
is encoded in gauge bridges between the
charges formed by the gauge fields.
The algebra of the electromagnetic field was therefore 
extended by them in~\cite{BuCiRuVa2022}; 
they are complementary to the local charge operators.    
This extended framework provided in addition to
the radiation also information about the energy contained in 
the gauge bridges. 

\medskip 
In the present article, we consider both approaches, 
but restrict the algebras to light cones.
We will show that this restriction cures
the infrared problems in the Minkowskian world. 
The radiation produced by 
the creation and transport of electric charges in light cones  
can be described by  photons in Fock space.
Longitudinal photons entering the gauge bridges
disappear rapidly and have little effect on inertial observers. 
So the notorious infrared problems 
in quantum electrodynamics have their origin in the assumption 
that observations take place in entire Minkowski space. 
Under the present assumptions about the
spacetime location of real experiments these problems disappear. 

\medskip
Our article is organized as follows.
We briefly summarize in Sect.~\ref{sec2} our algebraic framework and
the description of charge-carrying automorphisms, expounded in
\cite{BuCiRuVa2021,BuCiRuVa2022}. Section 3 contains 
the discussion of the radiation field in light cones
which arises by shifting charges to light-like infinity. 
We explain our arguments and present our results there. 
The analysis of the gauge bridges, appearing in the extended algebra, 
and the presentation of the
corresponding results is the subject of Sect.\ 4. Most 
proofs of our statements are given in appendices. 
The article ends with some brief comments on the
physical implications.

\section{Fields and charge creating automorphisms}
\label{sec2}
\setcounter{equation}{0}

We use a slight generalization of the
Gupta-Bleuler description of local vector potentials underlying
the electromagnetic field, cf.\ \cite{BuCiRuVa2021,BuCiRuVa2022}. It includes,
both, the observable electromagnetic field and the gauge fields
that mediate the presence of charges (gauge bridges). 
Let $\cD_1(\RR^4)$ be the space of real, vector-valued test functions
with compact support. The exponentials of the Gupta-Bleuler fields
are denoted by the symbols 
\be \label{e.2.1} 
e^{icA(f)} \, , \quad  c \in \RR \, , \, f \in \cD_1(\RR^4) \, .
\ee 
These symbols generate a group with an
involutive anti-automorphism $*$ (hermitian conjugation). 
It is fixed by the Weyl relations, $f,g \in \cD_1(\RR^4)$,  
\be \label{e.2.2} 
e^{iA(g)} e^{i cA(f)} = e^{i \, (c/2) \langle g, D f \rangle} \, e^{iA(g + c f)} \, , 
\quad \big( e^{iA(f)} \big) {}^* =  e^{- iA(f)} \, ,
\quad  e^{iA(0)} = 1 \, .
\ee
Here 
\be \label{e.2.3} 
\langle g, D f \rangle \coloneqq \int \! dx dy \, g^\mu(x) \, 
D(x-y) \, f_\mu(y) \, ,
\ee
where $D$ denotes the massless Pauli-Jordan commutator
function. Note that we do not assume that the vector potential
is a solution of the wave equation. But according to
the defining equations \eqref{e.2.2}, the exponentials
$e^{i A(\square f)}$ are elements of the center of the
group, where $\square$ denotes the d'Alembertian. 

\medskip
This group can be extended by a standard procedure to a
C*-algebra $\efW$, where the exponential functions of the vector
potential are represented by unitary operators,
cf.~\cite{BuCiRuVa2016}. The algebra $\efW$ is the inductive limit
of the family (net) of local subalgebras $\efW(\mO)$ which are
assigned to the open, bounded regions $\mO \subset \RR^4$.
Each $\efW(\mO)$ is generated by unitaries~$e^{iA(f)}$ 
with test functions $f$ having support in $\mO$.
On this net acts the proper orthochronous Poincar\'e group
$\mL_+ \! \! \rtimes \RR^4$ by automorphisms which induce
on the test functions $f$ entering in the generating unitaries the maps
\be
((\Lambda, x)f)(y) \coloneqq (\Lambda^{-1}f)(\Lambda^{-1}y-x) \, ,
\quad \Lambda \in \mL_+ \, , \,  x,y \in \RR^4 \, . 
\ee
One thereby arrives at a local, covariant
net of C*-algebras on Minkowski space.

\medskip
The electromagnetic field $F$ is represented by 
a subalgebra $\efV \subset \efW$, which is 
an example of the universal algebra 
considered in \cite{BuCiRuVa2016}. 
Since $F = d A$ (curl of $A$), it follows
from Poincar\'e's Lemma that 
$\efV$  is generated by the unitaries 
$e^{iA(h)}$, where $h$ are elements of the subspace
$\mC_1(\RR^4) \subset \mD_1(\RR^4)$ of functions
with vanishing divergence, $\delta h = 0$.
The field $F$ satisfies by construction the
homogeneous Maxwell equations and 
the divergence of its Hodge dual, given by $\delta d A$,   
enters the inhomogeneous Maxwell equation. There it 
is equated with the electric current, $\delta d A = J$,
which is sensitive to the charge distributions. 
It follows from equation~\eqref{e.2.2}
that the unitaries $e^{iJ(f)}$, $f \in \mD_1(\RR^4)$, are elements
of the center of~$\efV$. This fits with the present 
discussion of external charges, which are
not affected by the radiation field.
The subgroup generated by the radiation field
can thus be identified with the unitaries
$e^{iA(h)}$, $h \in \cC_1(\RR^4)$, modulo the center
originating from the current.   
  
\medskip
As is explained in \cite{BuCiRuVa2021, BuCiRuVa2022},
external charge distributions,
fixed by some signed measure $\delta m$, are described by outer automorphisms
$\beta_m$ of the observable algebra~$\efV$. They act on the
generating unitaries according to
\be \label{e.2.5}
\beta_m(e^{iA(h)}) \coloneqq e^{i \varphi_m(h)} e^{iA(h)} \, ,
  \quad h \in \mC_1(\RR^4) \, ,
\ee
where $\varphi_m$ is a real linear functional. 
We consider here functionals on $\cD_1(\RR^4)$
of the form 
\be \label{e.2.6} 
\varphi_m(f) \coloneqq - \int \! dx dy \, m^\mu(x) D(x-y) f_\mu(y) \, ,
\quad m \in \mD_1(\RR^4) \, .
\ee
The automorphisms $\beta_m$
are then given by the adjoint action of unitary operators
in $\efW$,
\be  \label{e.2.7} 
\beta_m(e^{iA(h)}) = e^{iA(m)} e^{iA(h)} e^{-iA(m)} \, .
\ee
In \cite{BuCiRuVa2022} the unitaries  $e^{iA(m)}$ were multiplied with
unitary operators, carrying an external charge,  so that the resulting operators
become gauge invariant. Since these charged operators do not affect
(commute with) the observables in~$\efV$, we can ignore them here. 
Evaluating the action of the automorphisms on the currents one obtains
\be  \label{e.2.8} 
\beta_m(e^{iJ(f)}) = e^{i \varphi_m(\delta d f)} e^{iJ(f)} \, , \quad
f \in \mD_1(\RR^4) \, .
\ee
It follows from equation \eqref{e.2.6} that this action differs 
from the identity  mapping if and only if $\delta m \neq 0$.
Since the spacetime integrals of the test functions
$\delta m$ vanish, the total charge transferred 
by these automorphisms is zero, however. 

\medskip 
In order to arrive at automorphisms that
transfer a non-trivial charge one has to consider 
bi-localized charge distributions, where 
the charge has some given non-zero value in a fixed region 
and the opposite value in a space-like separated region.
The location of charges
can precisely be determined by the electromagnetic field $F$,
making use of Gauss' law. Automorphisms that transfer a
total charge different from zero are obtained  
by shifting the compensating charge out of reach for possible 
observers. In~\mbox{\cite{BuCiRuVa2021, BuCiRuVa2022}} they were shifted
to space-like infinity. In the present article, 
where we consider observers in light cones, they
are moved to light-like infinity. The resulting
automorphisms are then no longer represented by 
the adjoint action of unitary operators in~$\efW$ and 
one must determine their limits. This is accomplished
by a study of the limits of the underlying
functionals $\varphi_m$. 

\medskip
Turning to the expectation functionals on the
algebra of observables $\efV$, we proceed from the
vacuum state $\omega_0$ that is fixed by 
\be \label{e.2.9} 
\omega_0(e^{iA(h)}) \coloneqq e^{(1/2) \langle h, D_+ h \rangle} \, ,
\quad h \in \mC_1(\RR^4) \, .
\ee
The exponent is defined by equation 
\eqref{e.2.3}, where the Pauli-Jordan commutator function~$D$ 
is replaced by its positive frequency part $D_+$.
By the GNS-construction one obtains from  
the vacuum state the familiar representation of the
electromagnetic field on Fock space, 
where the current vanishes, the Poincar\'e transformations are unitarily
implemented, and the energy is non-negative 
for all inertial observers. 

\medskip
States $\omega_m$ with a non-trivial charge distribution
$\delta m$ arise by composing the vacuum state with
the corresponding automorphisms,
$\omega_m \coloneqq \omega_0  \hspace*{1pt} \beta_m$.
They are pure states on $\efV$, such as the vacuum~$\omega_0$.
Since the current vanishes in the vacuum representation,
the representations induced by $\omega_m$ are disjoint
from the vacuum representation if $\delta m \neq 0$.
It is also clear that the spacetime translations cannot
be unitarily implemented in the resulting representations, 
because the translations act non-trivially on the current and 
change the expectation values of $\delta m$, whence
the representation. Nevertheless, 
one can determine the energetic effects of the
charge distributions on the electromagnetic field. There are
two issues, discussed in \cite{BuCiRuVa2022},
which we also address here:

\medskip \noindent
(I) \ Focusing on the energy of the radiation
field, one has to analyze the action of the automorphisms $\beta_m$
on the energy density of the electromagnetic
field in the vacuum representation. 
If \mbox{$m \in \mD_1(\RR^4)$}, the generators $H_m$,
obtained by integrating the resulting
densities over a Cauchy surface, are 
selfadjoint operators in the vacuum Fock space. In the corresponding rest
frame, they are given by the action of 
$\beta_{\bm^\perp} \coloneqq \mbox{Ad} \, e^{i\bA(\bm^\perp)}$ on 
the vacuum Hamiltonian~$H_0$, where
$\bm^\perp \coloneqq (0,\bm - \nabla \Delta^{-1} (\nabla \bm))
\in \cC_1(\RR^4)$ is the
transverse part of the spatial components of $m$,
cf.\ \cite[Sect.\ 4]{BuCiRuVa2022}.
With this input, the transition to charged automorphisms and their
impact on the radiation field can then be studied. To do this, the
properties of the automorphisms~$\beta_{\bm^\perp}$ must be controlled in the
limit in which the compensating charges are moved away.
The results enter in the analysis of the energetic and 
normality properties of the limit of the states
$\omega_{\bm^\perp} = \omega_0 \, \beta_{\bm^\perp}$.
As an aside, the Lorentz
transformations are spontaneously broken in these states
in Minkowski space, so that the specification of a rest system
is indicated. 

\medskip \noindent
(II) \ In order to obtain information about the energetic
properties of the gauge bridges 
that mediate Gauss's law, the algebra $\efV$ must be extended,
cf.\ the discussion in \cite[Sect.\ 3]{BuCiRuVa2022}.
One chooses again a rest frame and considers the
subalgebra  $\pmb{\efV} \subset \efW$
which is generated by the three spatial components $\bA$ of
the vector potential (temporal gauge \cite{St}). This subalgebra
is also generated by a net of local algebras on which 
the spacetime translations act covariantly,
but not the Lorentz transformations. The vacuum state is given by 
\be
\omega_0(e^{i \bA(\bg)}) \coloneqq e^{(1/2) \langle \bg, D_+ \bg \rangle} \, ,
\quad \bg \in \pmb{\cD}_1(\RR^4) \, ,
\ee
where $\pmb{\cD}_1(\RR^4) \subset \cD_1(\RR^4)$ is the subspace
of test functions with vanishing time component. 
One can recover in the corresponding representation the
observable algebra $\efV$ by identifying the electric field
with the time derivative $\dot{\bA}$ and the magnetic
field with the curl $\nabla \times \bA$. 
Local charge distributions, determined by $m \in \cD_1(\RR^4)$,  
are created by automorphisms $\beta_{\bn - \bm}$ of $\, \pmb{\efV}$,  
where $\bn \coloneqq \nabla \Delta^{-1}  \partial_0  m_0$.
These automorphisms are unitarily implemented in the vacuum
representation of $\pmb{\efV}$,
\be 
\beta_{\bn - \bm}(e^{i \bA(\bg)}) =
e^{ i \varphi_{\bn - \bm}(\bg)} \, e^{i \bA(\bg)}
= e^{i\bA(\bn - \bm)} e^{i \bA(\bg)} e^{-i\bA(\bn - \bm)} \, , \quad
\bg \in \pmb{\cD}_1(\RR^4) \, .
\ee
Just as in approach (I), 
one considers the automorphisms $\beta_{\bn - \bm}$ 
in the limit in which the compensating charge is
moved away. Their analysis provides the
information about the energetic and normality
properties of the (charged)
limits of the states $\omega_0  \beta_{\bn - \bm}$.

\section{Charges and radiation fields in light cones}
\label{sec3}
\setcounter{equation}{0}

We restrict now the observables to a forward light cone,
where we choose without loss of generality the open cone  
$V_+ \coloneqq \{ x \in \RR^4 : x_0 > | \bx | \}$
with apex at the origin $0$. The corresponding subalgebra
of observables is denoted by $\efV(V_+)$. For the
construction of bi-localized, overall neutral   charge distributions $m$
with support in $V_+$, we consider the bijective map
$\xi : \RR_+ \times \RR^3 \rightarrow V_+$, given~by 
\be \label{e.3.1}
  \xi(\tau, \bx) \coloneqq (\sqrt{\tau^2 + \bx^2}, \bx) \, , \quad 
  \tau > 0 \, , \ \bx \in \RR^3 \, . 
\ee 
It describes hyperboloids, called time shells,
which are labeled by the proper time $\tau$ of 
inertial observers who started from the origin. 
This map will be used to map  
time slices $0 < \tau_1 \leq \tau \leq \tau_2$ into  
regions between time shells in~$V_+$. Putting $x^2 \coloneqq x_0^2 - \bx^2 $,
the inverse map of $\xi$ is given by 
\be \label{e.3.2}
  \xi^{-1}(x) = \big( \sqrt{x^2}, \bx \big) \, , \quad
  x = (x_0, \bx) \in V_+ \, .
\ee
Singular functions (distributions) $m$ that enter in the vector potential  
between opposite charges, localized at points $x_1, x_2 \in V_+$ 
on a time shell $\tau$, are defined as follows. 
Connecting the pre-images $\xi^{-1}(x_1), \xi^{-1}(x_2)$ by a straight line,
their image under the action of $\xi$
defines a path from $x_1$ to $x_2$ on the
time shell. It is given by,
$0 \leq u \leq 1$, 
\be \label{e.3.3}
  u \mapsto w_{x_1, x_2}(u) \coloneqq
  \big(\sqrt{\tau^2 +
    (\bx_1 + u (\bx_2 - \bx_1))^2},
  \bx_1 + u (\bx_2 - \bx_1) \big) \, , 
\ee 
where $(\tau, \bx_i) = \xi^{-1}(x_i)$, $i=1,2$. The
components $m_{x_1, x_2}^\mu$, $\mu = 0, \dots , 3$, of 
$m_{x_1, x_2}$, entering in the vector potential 
between the charges are given by
\be \label{e.3.4}
x \mapsto  m^\mu_{x_1, x_2}(x) \coloneqq \int^1_0 \! du \, 
  \big( \tfrac{d}{du} w^\mu_{x_1, x_2}(u) \big) \, 
\delta(x-w_{x_1, x_2}(u)) \, .
\ee
By construction, they satisfy 
\be \label{e.3.5} 
\partial_\mu m^\mu_{x_1, x_2}(x)= - \int^1_0  \! du \, \tfrac{d}{du} 
\delta(x-w_{x_1, x_2}(u))= \delta(x-x_1) - \delta(x-x_2) \, .  
\ee
We proceed from the singular functions  $m_{x_1, x_2}$ to elements of
$\cD_1(\RR^4)$ by integration with suitable test functions, depending
on $x_1, x_2$. With reference to their pre-images, we
choose test functions of the specific form 
\be \label{e.3.6}
(\tau, \bx_1, \bx_2) \mapsto
\vartheta(\tau, \bx_1) \sigma(\bx_2 - \bx_1) \, ,
\ee
where $\vartheta$ has compact support in  
$\{ (\tau, \bx) : \tau >  0 \, , \, \bx \in \RR^3 \}$
and $\sigma$ has compact support in~$\RR^3$. We also assume that
for given $q \in \RR$
\be \label{e.3.7} 
\int \! d \tau d \bx \, \vartheta(\tau, \bx) = q \, , \quad
      \int \! d \bx \, \sigma(\bx) = 1 \, .
\ee
The resulting regularized functions (if their dependence
on the chosen test functions is omitted for the sake of simplicity) are 
\be \label{e.3.8}
x  \mapsto m^\mu_1(x) \coloneqq 
\int_0^1 \! du \! \int \! d\tau d\bx_1 d\bx_2 \, 
\vartheta(\tau,\bx_1) \sigma(\bx_2 - \bx_1) \, 
\big( \tfrac{d}{du} w^\mu_{x_1, x_2}(u) \big) \, 
\delta(x-w_{x_1, x_2}(u)) \, . 
\ee
Performing the integration with regard to
$\tau, \bx_1$ and making an obvious change
of variables,
this integral can be presented in the form
\be \label{e.3.9}
  x  \mapsto m^\mu_1(x) = 
\theta(x_0) \int_0^1 \! du \! \int \! d\by \, 
(1/\sqrt{x^2}) \, \vartheta\big(\sqrt{x^2},\bx - u \by \big) \,
\sigma\big(\by \big) \, 
\big((\bx \by) , x_0 \by \big)^\mu \, .
\ee
In view of the support properties of $\vartheta$ and $\sigma$, it is
apparent that $m^\mu_1$ has compact support in $V_+$. 
Since $x \mapsto \sqrt{x^2}$
is positive on $V_+$ and smooth, it follows that
$m_1 \in \cD_1(\RR^4)$. Moreover, 
relations \eqref{e.3.5} and \eqref{e.3.8} imply that
\be \label{e.3.10} 
x \mapsto \delta m_1(x) = \partial_\mu m^\mu_1(x) = \theta(x_0) 
(x_0 / \sqrt{x^2}) \, \big(\vartheta(\sqrt{x^2}, \bx) -
(\vartheta \star \sigma)(\sqrt{x^2}, \bx) \big) \, , 
\ee 
where $\vartheta \star \sigma$ denotes convolution of
the functions with regard to the spatial variables.

\medskip
In order to keep our discussion simple, we restrict
our attention to functions $\vartheta$ of the product form  
\be \label{e.3.11}
\tau, \bx \mapsto \vartheta(\tau, \bx)
\coloneqq \vartheta_0(\tau) \vartheta_1(\bx) \, ,
\ee
where $\vartheta_0$ has support in
$\{ \tau \in \RR :  0 < \tau_1 \leq \tau \leq \tau_2 \}$ and $\vartheta_1$
in $\{ \bx \in \RR^3 : |\bx| \leq r \}$.
For $\sigma$ we choose functions which have support in 
$\{ \bx \in \RR^3 : 0 < R \leq |\bx| \leq \bar{R}  \}$.
It then follows after a moments reflection that the supports of the
two charge densities on the right hand side of equation \eqref{e.3.10} 
are space-like separated if $R > (\tau_2^2 - \tau_1^2)/2\tau_1 +
(1 + \tau_2^2/\tau_1^2) \, r$. We will stick in the following
to functions with this product form and these support properties.

\medskip
The field created by pairs
of opposite charges in $V_+$ perturbs the electromagnetic field.
As was explained at the end of the preceding Sect.\ \ref{sec2} in
(I), its
energetic effects are described by a Hamiltonian
$H_m$ which is obtained
by the adjoint action of the unitaries $e^{i \bA(\bm^\perp)}$
on the vacuum Hamiltonian $H_0$.
The function $\bm^\perp$ is in general no longer
a test function because of the transverse projection.
But it is still an element of the single
particle space, so the unitary operators are well defined on the
vacuum Fock space. In other words, the representation of
the radiation field induced by the bi-localized charges is unitarily
equivalent to the vacuum representation in Minkowski space.
This changes, however, if the compensating charge is
shifted away.

\medskip 
In order to approximate charged states in $V_+$  with a fixed
charge distribution $\vartheta$ and to keep control
of the energy of the radiation field, we  
move the compensating charge to light-like infinity
by scaling $\sigma$,
\be \label{e.3.12}
\bx \mapsto \sigma_s(\bx) \coloneqq (1/s^3) \,
\sigma(\bx/s) \, ,
\quad s \geq 1 \, .
\ee
These data, plugged into equation \eqref{e.3.8}, yield 
test functions in $\cD_1(\RR^4)$ which are denot\-ed~by~$m_s$
(they are obtained by extension of the $u$-integral
from $1$ to $s$). The  
space-like distance of the supports of the corresponding 
charge densities in equation~\eqref{e.3.10} tends to infinity
in the limit of large $s$, while the  support of $m_s$ remains
in a hypercone between the time shells
$\tau_1 \leq \tau \leq \tau_2$, cf.\ \cite[Appx.]{BuRo2014}.
For the analysis of the
properties of the radiation field in this limit, we 
consider the four-dimensional Fourier transform
of the transverse component $\bm^\perp_s$ of $m_s$.
Using time-shell coordinates,
it is given by 
\begin{align} \label{e.3.13}
  p  \, & \mapsto \, \widetilde{\bm}_s^\perp(p) =
  ( \widetilde{\bm}_s - \bn (\bn  \widetilde{\bm}_s))(p) \nonumber \\
  & = \int \! d\tau d\bx \, e^{i(p_0 \sqrt{\tau^2 + \bx^2} - \bp \bx)} \, 
  \vartheta_0(\tau) \!
  \int_0^s \! \! du \! \! \int \! d\by \, 
\vartheta_1(\bx - u \by) \, \sigma(\by) (\by - \bn (\bn \by))
\, , 
\end{align}
where $\bn \coloneqq \bp / | \bp |$. It is shown in Appendix \ref{sectA}
that the restrictions of the functionals \mbox{$s \mapsto \varphi_{m_s}$} 
to test functions with support in $V_+$, 
cf.~\eqref{e.2.6}, converge for large~$s$.
Their limits~$\varphi_{m_\infty}$ involve the
distributions $m_\infty$ which are obtained by 
extending the $u$-integral in~\eqref{e.3.9} to infinity. 
Their transverse parts $\bm_\infty^\perp$
induce representations of the radiation field which are
charged, hence disjoint from
the vacuum representation in Minkowski space. 
Still, the spacetime translations are implemented 
in these representations by unitary operators which  
satisfy the relativistic spectrum condition. 
For the proof we make use of
a criterion by Roepstorff \cite{Ro}. The crucial step is to 
demonstrate that the limit distributions~$\widetilde{\bm}^\perp_\infty$
can be restricted to positive light-like momenta 
and form elements of the space~$L_1^ \perp(\RR^4)$ with (semi)norm given by
\be \label{e.3.14}
\| \bm^\perp \|^2_1 \coloneqq \int \! \frac{d \bp}{\sqrt{1 + |\bp|^2}} \,
|\widetilde{\bm}^\perp(p)|^2 \,  \Big|_{p_0 = | \bp |} \, .  
\ee
This space constitutes an extension of the massless single particle space
$L_0^\perp(\RR^4)$, where $1$ is replaced by $0$ in the square root in
equation \eqref{e.3.14}. 
It is shown in Appendix \ref{sectA} that 
\be \label{e.3.15} 
| \widetilde{\bm}^\perp_\infty(p)  | \, \Big|_{p_0 = |\bp|}
\leq c_k | \bp |^{-1 -k} \, , \quad 0 < k < 1/4 \, .
\ee
The constants $c_k$ depend, apart from the choice of $k$ within
the above limitations, on the underlying test functions
$\vartheta, \sigma$. This estimate implies that
$\bm^\perp_\infty \in L_1^\perp(\RR^4)$ and thus leads to the following
statement.
\begin{proposition}  \label{p.3.1}
Let $\bm^\perp_\infty$ be the radiation field obtained by 
generating pairs of opposite charges
and transporting one of them to light-like infinity.
The state $\omega_{\bm^\perp_\infty} \coloneqq \omega_0 \, \beta_{\bm^\perp_\infty}$ 
induces a representation of $\efV$ in which the spacetime translations are
unitarily implemented. The corresponding  generators satisfy
the relativistic spectrum condition (positivity of the energy
for inertial observers).
\end{proposition}

\vspace*{0mm}
Let us briefly discuss at this point what information
can be obtained in $V_+$ about the energy content
of radiation fields which are created by perturbing  
the vacuum. Note that the Minkowskian
Hamiltonian $H_0$ is not affiliated with the
(weak closure of the) algebra 
$\efV(V_+)$. Nevertheless, inertial observers
can obtain some information 
by using  the local stress energy tensor~$\theta_{\mu \nu}$,
which is a distribution-valued observable. Its integral over 
a Cauchy surface $\Sigma \subset V_+$, \eg a time shell,
is formally given by
$\int_\Sigma  d \, \Sigma^\mu(x) \, \theta_{\mu \, 0}(x)$.
Denoting by $\Omega$ the vacuum vector, one obtains for
a dense set of local observables $A \in \efV(V_+)$ 
\be \label{e.3.16} 
\langle \Omega,
[\int_\Sigma  d \, \Sigma^\mu(x) \, \theta_{\mu \, 0}(x) , A ] \Omega \rangle
=  \langle \Omega, [H_0 , A ] \Omega \rangle = 0 \, .
\ee
This equation is valid since the integral
extends only over a bounded region
because of the space-like commutativity of observables. So it
coincides with the integral over a locally deformed space-like
hyper-plane in $\RR^4$ and hence with $H_0$.
By standard arguments this equation implies that one arrives at 
a densely defined, symmetric operator $H_\Sigma$, putting 
\be
H_\Sigma A \Omega \coloneqq 
[\int_\Sigma  d \, \Sigma^\mu(x) \, \theta_{\mu \, 0}(x) , A ] \, \Omega \, .
\ee
It coincides with the restriction of the vacuum
Hamiltonian $H_0$ to that domain, which in fact 
is a core. Therefore, $H_0$ can be recovered in $V_+$
by subtracting from 
$\int_\Sigma  d \, \Sigma^\mu(x) \, \theta_{\mu \, 0}(x)$
the contributions resulting from vacuum polarization effects. 
The integral by itself does not lead to a meaningful operator, however, 
due to infinite fluctuations. 

\medskip
Having seen that the creation of charged states 
produces radiation fields with decent energetic properties, we will
show next that the states of these fields on light cone algebras
are normal relative to the vacuum state. This means that
the obstructions to normality considered possible in
\cite[Thm.\ 3.1]{CaDy} do not occur. Since the representations
induced by the radiation fields are covariant with regard to
spacetime translations, it suffices to establish
normality on a single light cone algebra.
Given $m_\infty$, we choose a light cone as follows.
The underlying function $m_1$ in \eqref{e.3.8} is a test function with  
compact support. So there exists some positive  
time translation $t$ such that
this support is contained in the backward light cone
$V_- + t$. We prove normality of the radiation field
on the corresponding forward light cone $V_+ + t$. This will 
be accomplished by making use of Huygens's principle.
It implies that the restriction of
the functional $\varphi_{m_\infty}$ in equation \eqref{e.2.6} to test 
functions with support in $V_+ + t$
does not change if one modifies 
$m_{\infty}$ in the past cone $V_- + t$. 
Taking advantage of this indeterminacy,  we will see that 
the radiation data, which are accessible
in $V_+ + t$, can be described by vectors in the
Fock space of photons.

\medskip 
Our argument is based on the 
fact that charges with constant velocity do not radiate.
We therefore return from $m_\infty$ to the
underlying gauge bridges between 
point charges, cf.~\eqref{e.3.4}. The
underlying paths \eqref{e.3.3} in $V_+$,
restricted to the region $V_+ \backslash (V_- + t)$, 
are then extended into $V_+ \cup (V_- + t)$ in such a way that
each of them approaches a light ray $u \mapsto  u \, l $
for large negative as well as positive $u \in \RR$. 
Radiation is caused by deviations from these light-like 
paths at intermediate times. For
charges shifted asymptotically on a time shell in light-like
directions, this radiation
is moderate and will be determined. Integrating
the resulting functions~\eqref{e.3.4} for the
extended paths with the test functions
$\vartheta_0, \vartheta_1, \sigma$
underlying~$m_\infty$, cf.~\eqref{e.3.8}, one obtains
a more regular function $m_{reg}$ with support
in $V_+ \cup (V_- + t)$. It coincides with $m_\infty$
on $V_+\backslash (V_- + t)$, hence
\mbox{$(\varphi_{m_\infty} - \varphi_{m_{reg}}) \upharpoonright \cD_1(V_+ + t) = 0$}.
Moreover, as is shown in Appendix~\ref{sectB}, the transverse 
part of $m_{reg}$ is element of the one-photon space, 
\mbox{$\bm^\perp_{reg} \in L_0^\perp(\RR^4)$}. As both,
$\bm^\perp_{reg}$ and $\bm^\perp_\infty$, induce 
representations which are covariant with regard to
the spacetime translations, the following result~obtains. 
\begin{proposition} \label{p.3.2} 
  Let $\omega_{\bm^\perp_\infty}$ be a state as in Proposition \ref{p.3.1} and 
  let $\omega_{\bm^\perp_{reg}}$ be the associated state in Fock space,
  described in the preceding discussion. The restrictions of these 
  states to the light cone algebra $\efV(V_+ + t)$ coincide  
  for some $t \in V_+$. 
  Whence, the restrictions are normal with regard to each other
  on all algebras $\efV(V_+ + x)$, $x \in \RR^4$. Thus, infrared clouds
  of photons cannot be detected in light~cones. 
\end{proposition}

We conclude this section by showing that 
the radiation field can be clearly distinguished from the other components of
the electromagnetic field by measurements in~$V_+$.
For the proof, we proceed from the 
three-vector-valued test functions \mbox{$\bfy \in \, \pmb{\cD}_1(V_+)$}
having compact support in the (open) light-cone $V_+$. The derived
functions $h \coloneqq \big( 0, -\Delta \bfy + \nabla (\nabla \bfy) \big)$
are elements of $\cC_1(V_+)$ that lie in a 
subspace of wave functions in $L_0^\perp(\RR^4)$ fixed by the
chosen inertial system. In fact, they 
are dense in $L_0^\perp(\RR^4)$. To prove this, note that 
the functions $x \mapsto h_t(x) \coloneqq h(x_0 - t, \bx)$ are
elements of $\cC_1(V_+)$ for all~$t \geq 0$. So, by
a double Laplace-transformation, 
$\varepsilon \mapsto  h^\varepsilon \coloneqq 
\int_0^\infty \! ds \int_0^\infty \! dt \, e^{- \varepsilon (s + t)} h_{s + t}$, 
$\varepsilon > 0$, one obtains test functions with support in~$V_+$. Their 
Fourier transforms are 
Schwartz test functions (omitting constant factors) 
\be \label{e.3.18}
p \mapsto \widetilde{h}^\varepsilon(p) = 
(p_0 - i\varepsilon)^{-2} \, \widetilde{h}(p) =
(p_0 - i\varepsilon)^{-2}  \big( 0, \bp^2 \widetilde{\bfy}(p) -
\bp \, (\bp \widetilde{\bfy}(p)) \big)
\, , \quad
\varepsilon >  0 \, .
\ee
For lightlike momenta 
$p_0 = | \bp |$ they converge in the
one-photon space $L_0^\perp(\RR^4)$ 
to~$\bfy^\perp$ in the limit of small $\varepsilon$. 
That they are dense in $L_0^\perp(\RR^4)$ follows from the Reeh-Schlieder
theorem.
Moreover, in view of the decay properties
of~$p \mapsto \widetilde{\bfy}(p)$
for large momenta, it is also
clear that for $m_\infty^\perp \in L_1^\perp(\RR^4)$
the functions $h^\varepsilon \mapsto \varphi_{m_\infty}(h^\varepsilon)$ 
converge to $\varphi_{m_\infty}(\bfy^\perp)= \varphi_{m_\infty^\perp}(\bfy)$ on the
test functions~\eqref{e.3.18} in the limit
of small~$\varepsilon$.
It implies that the transverse part of the
electromagnetic field and the dispersion
law of photons can be determined with great
accuracy in bounded spatial regions
by extended temporal measurements. 

\section{Charges and gauge bridges in light cones}
\label{sec4}
\setcounter{equation}{0} 

We turn now to the analysis of the observable
features of the gauge bridges which connect the charges. 
As outlined in Section \ref{sec2}, we consider 
the extended algebra $\pmb{\efV} \subset \efW$, 
which describes in addition to the two transverse
degrees of freedom of the photons the longitudinal degree
of the gauge bridges. Consistent with the preceding
notation, we consider now 
three-vector-valued wave functions 
$\bl \in \pmb{\cD}_1(\RR^4)$.
The spaces spanned by them are denoted
by~$L_\iota(\RR^4) \supset L_\iota^\perp(\RR^4)$
with (semi)norms given by 
\be \label{e.4.1}
\| \bl \|^2_\iota \coloneqq \int \! \frac{d \bp}{\sqrt{\iota + |\bp|^2}} \,
|\widetilde{\bl}(p)|^2 \,  \Big|_{p_0 = | \bp |} \, ,
\quad \iota = 0,1 \, .
\ee
The exponentials $e^{i \bA(\bl)}$ for $\bl \in L_0(\RR^4)$ are
unitary operators in the vacuum representation whose adjoint
action results in the automorphisms  $\beta_{\bl}$. 
If $\bl \in L_1(\RR^4)$, the corresponding action $\beta_{\bl}$ defines 
outer automorphisms of the algebra ${} \, \pmb{\efV}$,  which in general are
not unitarily implemented in the vacuum representation. In those cases 
the states $\omega_0 \, \beta_{\bl}$ are disjoint from 
the states in the vacuum representation. But one still
has the following variant of the result by Roepstorff \cite{Ro}, 
mentioned in Sect.\ \ref{sec3}. Its proof is given in Appendix~\ref{sectC}. 
\begin{lemma} \label{l.4.1}
  Given $\bl \in L_1(\RR^4)$, let $\omega_0 \, \beta_{\bl}$ be the
  state on $\pmb{\efV}$ that describes the corresponding
  gauge bridge and radiation field. 
  The spacetime translations $\RR^4$ are unitarily
  implemented in the representation induced by this state with 
  generators that satisfy the relativistic spectrum condition. 
\end{lemma}  
For the construction of charged states, we return to 
the functions $m_s \in \cD_1(\RR^4)$, which are
obtained by extending the $u$-integral
in equation \eqref{e.3.8} to $s \geq 1$. As was explained at
the end of Sect.~\ref{sec2}, local charge distributions are induced by  
automorphisms $\beta_{\bn_s - \bm_s}$ on $\, \pmb{\efV}$. 
In the vacuum representation they are given by the adjoint
action of the unitaries $e^{i \bA(\bn_s - \bm_s)}$, where both  
$\bm_s$ (the spatial components of $m_s$)  and
$\bn_s = \nabla \Delta^{-1} \partial_0 \, m_{s \, 0}$
are elements of $L_0(\RR^4)$. We split
\be \label{e.4.2} 
(\bn_s - \bm_s)  
= \nabla \Delta^{-1}(\partial_0 \, m_{s \, 0} - \nabla \bm_s)
-(\bm_s - \nabla \Delta^{-1} \nabla \bm_s)
=  \nabla \Delta^{-1} \delta m_s - \bm_s^\perp \, .
\ee
Making use of the Weyl relations \eqref{e.2.2}, the automorphisms
$\beta_{\bn_s - \bm_s}$ factorize,
\be
\beta_{\bn_s - \bm_s} = \beta_{- \bm_s^\perp} \, 
\beta_{\nabla \Delta^{-1} \delta m_s} \, , \quad s \geq 1 \, .
\ee
As explained in Sect.~\ref{sec3},
the functions $\bm_s^\perp$ converge in the limit
of large $s$ to  \mbox{$\bm^\perp_\infty \in L_1(\RR^4)$}.
Hence the corresponding automorphisms
$\beta_{-\bm^\perp_s}$ converge pointwise in
norm on~$\, \pmb{\efV}$. 
The automorphisms $\beta_{\, \nabla \Delta^{-1} \delta m_s}$
(inducing the gauge bridge between the charges) 
also converge, as is shown 
in the subsequent lemma. There we make use of
relation \eqref{e.3.10} according to which
$\delta m_s$ is the difference between the charge distributions
at the initial and final points of the
underlying path, \ie $\delta m_s = \varrho_0 - \varrho_s$.
\begin{lemma} \label{l.4.2}
  The automorphisms $s \mapsto \beta_{\, \nabla \Delta^{-1} \delta m_s}$ 
  converge pointwise in norm on the alge\-bra $\pmb{\efV}$ when $s$
  tends to infinity. Their limit is 
  $\beta_{\, \nabla \Delta^{-1} \varrho_0}$, where
  $\nabla \Delta^{-1} \varrho_0 \in L_1(\RR^4)$.
  \end{lemma}  
\begin{proof}
  Since $\varrho_0, \, \varrho_s$ are test functions,
  $\nabla \Delta^{-1} \varrho_0$ and $\nabla \Delta^{-1} \varrho_s$
  are contained in $L_1(\RR^4)$. So one can factorize the automorphism, 
  $\beta_{\, \nabla \Delta^{-1} \delta m_s} =
  \beta_{\, \nabla \Delta^{-1} \varrho_0} \, \beta_{\, \nabla \Delta^{-1} \varrho_s}^{-1}$.
  We must show that the second factor converges to the identity.
  According to relations \eqref{e.2.5} and \eqref{e.2.6},
  this is accomplished by showing that the functionals
  \be \label{e.4.3}
  s \mapsto \varphi_{\, \nabla \Delta^{-1} \varrho_s}(\bg)
  = \varphi_{\Delta^{-1} \varrho_s}(\nabla \bg) \, , 
  \quad \bg \in \, \pmb{\cD}_1(\RR^4) \, ,
  \ee
  converge to $0$. The Fourier transform of $\varrho_s$ is given by
  (omitting constant factors)
  \be \label{e.4.4}
  p \mapsto \widetilde{\varrho_s}(p) =
  \int \! d\tau \vartheta_0(\tau) \int \! d\bx_1 d\bx_2
   \, \vartheta_1(\bx_1) \sigma(\bx_2 - \bx_1/s)
  e^{i s (p_0 \sqrt{\tau^2/s^2 + \, \bx_2^2 \, } - \bp \bx_2)} \, .
  \ee
Whence, the functional \eqref{e.4.3} is equal to the imaginary part of 
\be \label{e.4.5}
\hspace*{-2mm}
\int \! \! \frac{d \bp}{2 |\bp|^3} \, (\bp \, \widetilde{\bg})(p) \! \! 
\int \! d\tau \vartheta_0(\tau) \! \!  \int \! d\bx_1 d\bx_2
\, \vartheta_1(\bx_1) \sigma(\bx_2 - \bx_1/s)
e^{-i s (p_0 \sqrt{\tau^2/s^2 \, + \, \bx_2^2 \, } - \bp \bx_2)}  \, 
\Big|_{p_0 = | \bp |}.
\ee
Since all functions involved are test functions,
this integral is absolutely integrable. 
In view of the support properties of $\vartheta_0$, $\vartheta_1$, and
$\sigma$, a simple estimate shows that for large~$s$ the terms 
$\bx_1/s$ in the argument of $\sigma$ and $\tau^2/s^2$ in the square
root of the exponential can be neglected.
Recalling the normalization \eqref{e.3.7}, the expression~\eqref{e.4.5}
converges to  
\be \label{e.4.6}
q \int \! \frac{d \bp}{2|\bp|^3} \, (\bp \, \widetilde{\bg})(p) \! 
\int \! d\bx_2 \, \sigma(\bx_2) e^{i s ( p_0 |\bx_2| - \bp \bx_2)}  \, 
\Big|_{p_0 = | \bp |} \, .
\ee
The second (unrestricted) integral is a solution of the wave equation
in $p$. Since $\sigma$ is a test function which vanishes in
a neighborhood of $0$, its Cauchy data are test functions. So the modulus
of the integral is bounded
in light-like directions $p_0 = | \bp |$
by $c \, (1 + s |\bp|)^{-1}$, cf.\ \cite[Thm.\  XI.19(c)]{ReSi}. 
It is then clear that the expression \eqref{e.4.6} vanishes if $s$ tends to
infinity, completing the proof.
\end{proof}
Since the functions $\nabla \Delta^{-1} \varrho_0$ and
$\bm^\perp_\infty$ are both elements of $L_1(\RR^4)$, the
following statement is an immediate consequence of Lemma \ref{l.4.2}.
\begin{proposition} \label{p.4.3}
  Let $ (\nabla \Delta^{-1} \varrho_0 - \bm^\perp_\infty) \in L_1(\RR^4)$
  be the gauge
  bridge and radiation field obtained by transporting 
  in a pair of opposite charges in $V_+$
  one of the charges to light-like infinity within a hypercone.
  The charged limit
  state $\omega_0 \, \beta_{(\nabla \Delta^{-1} \varrho_0 - \bm^\perp_\infty)}$
  on $\pmb{\efV}$ induces a representation in which the spacetime
  translations are unitarily implemented. Its generators
  satisfy the relativistic spectrum condition. 
\end{proposition}  
It remains to establish the  normality of the
states $\omega_0 \, \beta_{(\nabla \Delta^{-1} \varrho_0 - \bm^\perp_\infty)}$
relative to the vacuum representation on all 
subalgebras $\, \pmb{\efV}(V_+ + x)$, $x \in \RR^4$.
We again use the factorization of the automorphisms,
$\beta_{(\nabla \Delta^{-1} \varrho_0 - \bm^\perp_\infty)}
=  \beta_{- \bm^\perp_\infty}  \beta_{\nabla \Delta^{-1} \varrho_0}$, 
and the covariance of the resulting 
representation of $\, \pmb{\efV}$ with respect to spacetime translations.
Thus we need to establish normality only for a single
light cone algebra. As in the preceding Sect.~\ref{sec3}, 
we choose a positive time translation $t$ such that 
the support of $m_1$ in equation~\eqref{e.3.9} (hence that
of $\nabla \varrho_0$) lies in the backward light cones
$V_- + t$. The translated test functions
\be
x \mapsto \nabla \varrho_{0 , s}(x) \coloneqq \nabla \varrho_0(x_0 - s, \bx) \, ,
\quad s \leq 0 \, ,
\ee
then also have support in $V_- + t$.
Thus, as was explained at the end of Sect.~\ref{sec3},
one obtains by a double Laplace transformation,
extending now over negative times, a test function with
support in $V_- + t$. Its Fourier transform is (apart from 
a constant factor) 
\be
p \mapsto (p_0 + i \varepsilon)^{-2} \, \bp \, \widetilde{\varrho}_0(p) \, . 
\ee
Restricting it to light-like momenta, $p_0 = | \bp |$, it is
apparent that it converges strongly in
$L_1(\RR^4)$ to $\nabla \Delta^{-1} \varrho_0$
in the limit of small $\varepsilon$. 
Taking Huygens's principle into account, this implies that the functional
$\bg \mapsto \varphi_{\nabla \Delta^{-1} \varrho_0}(\bg)$ 
vanishes on all test functions \mbox{$\bg \in \pmb{\cD}_1(V_+ + t)$}.
Hence $\beta_{\nabla \Delta^{-1} \varrho_0}$ acts like the identity 
on the algebra $\, \pmb{\efV}(V_+ + t)$. It
shows that effects of the gauge bridges observable in $V_+$ disappear
completely within a limited time $t$ in the complement of $V_+ + t$.

\medskip
A direct consequence of this observation is that the particle
properties of ``longitudinal photons'' (\eg their dispersion law)
cannot be accurately determined by inertial observers, in contrast
to those of real photons.  This casts light on the statement
in textbooks that longitudinal photons are to be regarded as
``virtual particles''. 

\medskip 
Turning to the radiation field induced by $\beta_{- \bm^\perp_\infty}$,
we have shown in Sect.~\ref{sec3} that
\be
(\bm^\perp_\infty - \bm^\perp_{reg})
\upharpoonright V_+ \backslash (V_- + t)
= 0 \, ,
\ee
where $\bm^\perp_{reg} \in L_0(\RR^4)$. Hence 
\be
(\beta_{-\bm^\perp_\infty} - \beta_{-\bm^\perp_{reg}}) \upharpoonright
\pmb{\efV}(V_+ + t) = 0 \, ,
\ee
where 
$\beta_{-\bm^\perp_{reg}}$ is implemented in the  vacuum representation of
$\, \pmb{\efV}$ by the adjoint action of unitary operators.
So, in summary, we have
\be
( \omega_0 \beta_{(\nabla \Delta^{-1} \varrho_0 - \bm^\perp_\infty)}^{}
- \omega_0 \beta_{-\bm^\perp_{reg}} ) \upharpoonright \pmb{\efV}(V_+ + t) = 0 \, , 
\ee
where $\omega_0 \beta_{-\bm^\perp_{reg}}$ is a state in the vacuum
representation of $\ \pmb{\efV}$. This is proof of the following statement.
\begin{proposition} \label{p.4.4}
  Let $\omega_0 \beta_{(\nabla \Delta^{-1} \varrho_0 - m_\infty^\perp)}$
  be a charged state as in Proposition \ref{p.4.3}. Its restriction to
  any light cone algebra $\ \pmb{\efV}(V_+ + x)$, $x \in \RR^4$, is
  normal relative to the vacuum representation of $\, \pmb{\efV}$.
  In particular, the gauge bridges emanating from charged states
  cannot be distinguished in light cones from those
  appearing in the vacuum representation. Furthermore,
  their impact disappears for inertial observers within an
  instant.
\end{proposition}  
\section{Conclusions} \label{sect5}
\setcounter{equation}{0}

In this article, we have analyzed the impact of electric
charges on the electromagnetic field that can be tested and perceived by
real observers. In doing so, we have taken seriously the fact that
experiments can only be performed in parts of Minkowski space, at best
within future directed light cones.
To simplify the analysis, we have made the assumption that the charges
are external in nature, \ie they act on the electromagnetic field, but
there is no back-reaction. One can then consistently assume that the
charges are compactly localizable. (In case of mutual
interactions this is not possible \cite{BuDoMoRoSt}.) Furthermore,
the gauge bridges connecting the charges can be described
by hypothetical particles with discrete mass zero (longitudinal photons).
We believe that the findings obtained in the present analysis do
not depend significantly on these idealizations.
In particular, the remarkable qualitative differences between
the theory restricted to light cones and the global Minkowskian
theory appear to be of general validity.

\medskip
First, the production of charged states from neutral
bi-localized states in a light cone requires that one of the
charges be moved to light-like infinity,
\ie into the space-like complement of the observers.
This compensating charge must be brought to the speed
of light, which can certainly only be achieved approximately
in view of the energy required. Nevertheless, the energy
of the electromagnetic field created by this operation
is well defined and 
stays bounded below also in the limit, as we have shown.
In Minkowski space, one frequently discusses  
that the compensating charge is moved to
space-like infinity. This costs little if no energy.
But this procedure is only possible in theory and can not be
accomplished by real operations, not even approximately.

\medskip
Second, we have seen that the notorious infrared clouds
of massless particles predicted by Minkowskian theory in the presence
of charges do not occur in light cones. All modes of the radiation
field that can be resolved in light cones can be described by states
in the Fock space of the photons. Since the Lorentz transformations
are unitarily implemented in Fock space and light cone algebras are
stable under their adjoint action, the restriction of the theory to
a light cone induces a Lorentz-invariant infrared cutoff. In this manner,
the infrared radiation is quenched geometrically.

\medskip
Finally, we have shown that the longitudinal photons that give rise to
Gauss's law in the gauge bridges between the charges can hardly be observed.
Their starting points are the trajectories of the charges.
From there they travel at the speed of light. They therefore disappear
quickly in the space-like complement of inertial observers. This makes it
impossible to determine their dispersion law with precision.
Although they are particles with a discrete mass, they elude systematic
observation and are virtual in this sense. Without engaging
in speculation, it seems possible that such
elusive particles are part of nature. 

\begin{appendix}

\section{Appendix} \label{sectA} 
\setcounter{equation}{0}

Here we carry out the missing steps
in the proof of Proposition \ref{p.3.1}. Starting point
is the sequence of functionals $s \mapsto \varphi_{m_s}$, cf.~\eqref{e.2.6}.
The functions $m_s$ are given by \eqref{e.3.9}, where the integral 
with respect to~$u$ extends from $0$ to any $s \geq 1$.
We will establish the convergence
of the functionals in the limit of large~$s$ on the subspace of
three-vector-valued test functions~$\, \pmb{\cD}_1(V_+)$.
The transverse parts $\, \pmb{m}_\infty^\perp$ of the underlying limit
distributions 
$m_\infty$ will be shown to satisfy relation
\eqref{e.3.15}; hence they are elements of $L_1^\perp(\RR^4)$.
As outlined in Sect.\ \ref{sec3}, the statement then follows from
results of Roepstorff \cite{Ro}. 

\medskip
Since the underlying charge distributions are assumed to have
the product form \eqref{e.3.11} and the 
functionals $\varphi_{m_s}$ are integrated with
three-vector-valued test functions,
we only need to consider the spatial components of $m_s$,
\be \label{e.a.1}
x \mapsto \bm_s(x) = \theta(x_0) (x_0/\sqrt{x^2}) \, \vartheta_0(\sqrt{x^2}) \, 
\int_0^s \! du \! \int \! d\by \, \vartheta_1(\bx - s \by) \sigma(\by) \by \, .
\ee
Because of the choice of $\vartheta_0$, $\vartheta_1$, and $\sigma$, 
the derivatives $\partial_s m_s$ are test functions
which have for given $s \geq 1$ support in the region 
\be \label{e.a.2}
\{ x \in \RR^4 : x_0 > 0, \ 0 < \tau_1 \leq \sqrt{x^2} \leq \tau_2 \, , \ 
sR - r \leq | \bx | \leq s\overline{R} + r \} \, .
\ee
For each compact region $\mO$ in the (open) light-cone $V_+$, the
former region lies in the spacelike complement of $\mO$  
if $s$ is sufficiently large. 

\medskip
Now, given any 
$\bfy \in \pmb{\cD}_1(V_+)$ with support in $\mO$,
its convolution with the commutator function, $x \mapsto (D \bfy)(x)$, 
is smooth and has support in the union of the forward and
backward light-cones having their apex in $\mO$. Hence the
function $s \mapsto \int \! dx \, \partial_s m_s(x) (D \bfy)(x)$
vanishes for large $s$. This implies that
$s \mapsto \int_0^s \! du \,  \int \! dx \, \partial_u \bm_u(x) (D \bfy)(x)
= \varphi_{\bm_s}(\bfy)$ is constant for large $s$.
It is proof of the statement that the functionals
$s \mapsto \varphi_{\bm_s}$ converge in the limit of
large $s$ on $\, \pmb{\cD}_1(V_+)$ to a functional
$\varphi_{\bm_\infty}$. 

\medskip
For the analysis of the transverse component $\bm_\infty^\perp$ of the
limit distribution $m_\infty$, we proceed to 
momentum space. Making use of time shell coordinates, the     
Fourier transform of~$\bm_\infty^\perp$
is given by (omitting constant factors)
\be \label{e.a.3}
p  \mapsto \widetilde{\bm}_\infty^\perp(p) =
\int \! d\tau d\bx \, e^{i(p_0 \sqrt{\tau^2 + \bx^2} - \bp \bx)}
\vartheta_0(\tau) \brho_{\bn}(\bx) \, , 
\ee
where 
\be \label{e.a.4}
\bx \mapsto \brho_{\bn}(\bx) \coloneqq 
\int_0^\infty \! \! \! du \! \! \int \! d\by \, 
\vartheta_1(\bx - u \by) \, \sigma(\by) (\by - \bn (\bn \by)) \, ,
\quad \bn \coloneqq \bp/| \bp |  \, .
\ee
Given $\bx \in \RR^3$, one sees by taking into account the support properties
of $\vartheta_1$ and $\sigma$ 
that the latter integral extends over the compact region
\be 
\{ (u, \by) \in \RR_+ \times \RR^3 :
R \leq |\by| \leq \bar{R} \, , \  |\bx - u \by| \leq r \} \, .
\ee
The derivatives of $\brho_{\bn}$ with regard to the components
of~$\bx$ can thus be interchanged with the integrals, where they act
on the test function $\vartheta_1$. Hence $\brho_{\bn}$ is smooth.
For the determination of the asymptotic properties of $\brho_{\bn}$, let
$|\bx| > r$. The integral 
then extends over the spatial region  
$\{ \by \in \RR^3 :
\by \bx \geq | \by | \sqrt{\bx^2 - r^2} \, , \
R \leq | \by | \leq \overline{R} \}$
and the interval
\be
\{ u \in \RR_+ : | u \, \by^2 - (\by \bx) | \leq
\sqrt{ (\by \bx)^2 - \by^2(\bx^2 - r^2) } \} \, .  \tag{A.6} 
\ee
The length of this interval is bounded by 
$2 \sqrt{(\bex_y \bx)^2 - (\bx^2 - r^2)}/R$,
where $\bex_y \coloneqq \by / | \by | $, 
and its lower edge is bigger than $(| \bx | - r)/\bar{R}$.
Based on this information on 
the integrand in relation \eqref{e.a.4}, one obtains  
an upper bound of $| \brho_{\bn} |$ if one 
replaces $\vartheta_1$ and $\sigma$
by the supremum of their modulus. Performing the 
$\by$-integration in spherical
coordinates and the $u$-integration, results in the  estimate 
\be \label{e.a.7} 
| \brho_{\bn}(\bx) | \leq c
\int \! d \hspace{1pt}
\Omega_{\bex_x}(\bex_y) \, \sqrt{(\bex_y \bx)^2 - (\bx^2 - r^2) } \, 
\sqrt{1 - (\bex_y \bn)^2} \, ,  \tag{A.7}
\ee
Here $d \hspace{1pt} \Omega_{\bex_x}(\bex_y)$ denotes the spherical measure 
restricted to $\{ \bex_y : \bex_y \bex_x \geq \sqrt{1 - r^2/\bx^2} \, \}$, 
where $\bex_x \coloneqq \bx / | \bx |$,
and the constant $c$ involves the product of the supremum of the
underlying test functions
as well as the contributions of the radial integration.
Now
\be  \label{e.a.8} 
(1 - (\bex_y \bn)^2)  \leq  (1 - (\bn \bex_x)^2) +
\sqrt{(1- (\bex_y \bex_x)^2 )} 
\leq (1 - (\bn \bex_x)^2) + r/|\bx| \, .  \tag{A.8} 
\ee
Replacing the second square root in \eqref{e.a.7} by the
square root of this upper bound, a straightforward computation of
the resulting integral yields 
\be  \label{e.a.9}
|\brho_{\bn}(\bx)|
\leq c \big( \sqrt{1 - (\bn \bex_x)^2} + r/|\bx| \big)^{1/2}/|\bx|^2 \, ,
\quad |\bx| \geq r \, .  \tag{A.9}
\ee

\medskip
We also need estimates of the 
Cauchy-Euler derivative $(\bx \bnabla_{\bx})$ of~$\brho_{\bn}$ on $\RR^3$, 
\begin{align} \label{e.a.10}
  (\bx \bnabla_{\bx})
  \brho_{\bn}(\bx) &  = \int_0^\infty \! \! \! du \! \! \int \! d\by \, 
(\bx \bnabla_{\bx}) 
  \vartheta_1(\bx - u \by) \, \sigma(\by) (\by - \bn (\bn \by)) \nonumber \\
  & = \sum_{j = 1}^3 x_j \int_0^\infty \! \! \! du \, (1/u) \! \!
  \int \! d\by \, \vartheta_1(\bx - u \by)
  \big( \partial_{y_j} \sigma(\by) \big) (\by - \bn (\bn \by))
  \nonumber \\
  & \, + \big(\bx - \bn (\bn \bx) \big)
  \int_0^\infty \! \! \! du \, (1/u)
  \! \! \int \! d\by \, \vartheta_1(\bx - u \by)
  \, \sigma(\by) \, . \tag{A.10}
\end{align}
As in the preceding step, one replaces the integrands on the
right hand side of this equality by their modulus. The factor 
$(1/u)$ can then be estimated by its value at the lower
edge of the integration interval, determined by the
test functions. If
$| \bx | \geq 2r $, it results in a factor $2/|\bx|$
that compensates for the increasing factors
in front of the integrals. 
The remaining integrals can be estimated as in the preceding
step, showing that the Cauchy-Euler derivative of~$\brho_{\bn}$ satisfies a 
bound as in \eqref{e.a.9}. This applies also to 
higher powers of the derivative by repeating 
this computation.

\medskip
We can turn now to the
estimate of $\widetilde{\bm}_\infty^\perp$ and establish the bound  
given in relation~\eqref{e.3.15}. To this end we 
use the Cauchy-Euler derivative on $\RR^4$, 
$D \coloneqq \tau \, \partial_\tau + (\bx \bnabla_{\bx})$, 
and consider for given $\bn = \bp/|\bp|$ the function
$\tau, \bx \mapsto \eta_{\bn}(\tau, \bx)
\coloneqq (\sqrt{\tau^2 + \bx^2} - \bn \bx)$.
For any differentiable function $\zeta: \RR \rightarrow \CC$,
one has
$D \, \zeta(\eta_{\bn}(\tau, \bx)) =
\eta_{\bn}(\tau, \bx) \, \zeta'(\eta_{\bn}(\tau, \bx))$,
where $\zeta'$ denotes the derivative. Thus 
\begin{align} \label{e.a.11}
  & |\bp|^2 \! \int \! d\tau d\bx \, e^{i |\bp| \eta_{\bn}(\tau,\bx) } \, 
\vartheta_0(\tau) \brho_{\bn}(\bx) \nonumber \\
& =  \int \! d\tau d\bx \,
\big( (\eta_{\bn}(\tau, \bx)^{-1} D)^2 \,
(1 - e^{i |\bp| \eta_{\bn}(\tau, \bx)}) \big) \, 
\vartheta_0(\tau) \brho_{\bn}(\bx) \nonumber \\
& =  \int \! d\tau d\bx \,
\big((D+1)(D +2) \, \eta_{\bn}(\tau, \bx)^{-2}
(1 - e^{i |\bp| \eta_{\bn}(\tau, \bx)}) \big) \,
\vartheta_0(\tau) \brho_{\bn}(\bx) \nonumber \\
& =  \int \! d\tau d\bx \,
\eta_{\bn}(\tau, \bx)^{-2} (1 - e^{i |\bp| \eta_{\bn}(\tau, \bx)})
(2 + D)(3 + D) \vartheta_0(\tau) \brho_{\bn}(\bx) \, . \tag{A.11}
\end{align}
The first two equalities arise by the action
of $D$ on functions of $\eta_{\bn}$ and 
the last equality obtains by partial integration,
where $D$ is converted to 
$D^T \coloneqq -(\partial_\tau \tau + (\bnabla_{\bx} \bx))
= -(4 + D)$.  
The resulting function
$\tau, \bx \mapsto (2 + D) (3 + D) \vartheta_0(\theta)
\brho_{\bn}(\bx)$ is a finite sum of products of
Cauchy-Euler derivatives   
of $\vartheta_0$ and $\brho_{\bn}$. Thus the
estimate \eqref{e.a.9} can be applied to all terms 
in this sum of integrals. By some abuse of notation,
the corresponding integrands have the form 
\be \label{e.a.12}
\tau, \bx \mapsto \eta_{\bn}(\tau, \bx)^{-2}
(1 - e^{i |\bp| \eta_{\bn}(\tau, \bx)})  
{\vartheta}_0(\tau) {\brho}_{\bn}(\bx) \, .   \tag{A.12} 
\ee
Picking $l$, $3/4 < l < 1$, they are bounded by
\be  \label{e.a.13}
|\eta_{\bn}(\tau, \bx)^{-2} (1 - e^{i |\bp| \eta_{\bn}(\tau, \bx)})  
{\vartheta}_0(\tau) \brho_{\bn}(\bx)|
\leq 2 | \bp |^l | \eta_{\bn}(\tau, \bx)|^{-2+l} |\vartheta_0(\tau)|
|\brho_{\bn}(\bx)| \, . \tag{A.13} 
\ee
Since $| \eta_{\bn}(\tau, \bx)| ^{-1}
\leq (\sqrt{\tau^2 + \bx^2} + |\bx|)/\tau^2$
and the $\tau$-integration extends over the interval
$0 < \tau_1 \leq \tau \leq \tau_2$, the upper bound is absolutely
integrable with regard to $\tau, \bx$
if $|\bx| \leq 2r$. In the region  $|\bx| \geq 2r$ 
one has 
\begin{align}  \label{e.a.14}
  & | \bx | | \eta_{\bn}(\tau, \bx) |^{-1}
  \big( \sqrt{1 - (\bn \bex_x)^2} + r/|\bx| \big)^2
  \tag{A.14}  \\
& \leq \sqrt{4 + \tau^2/r^2} \big(1 - (\bn \bex_x)^2 + \tau^2/\bx^2 \big)^{-1}
\big(1 - (\bn \bex_x)^2 + r^2/\bx^2 \big)
\leq 2 (\tau^2 + r^2)^{3/2}/(\tau r^2) \, ,  \nonumber
\end{align}
\ie the function on the left hand side
is bounded for $\tau_1 \leq \tau \leq \tau_2$
and $|\bx| \geq 2r$. Taking the fourth root
of this inequality and making use of the estimate \eqref{e.a.9}, one can
proceed from \eqref{e.a.13} to 
\be  \label{e.a.15}
| \eta_{\bn}(\tau, \bx)|^{-2+l} |\vartheta_0(\tau)| |\brho_{\bn}(\bx)|
\leq c(\tau,r) \, ( |\bx| \eta_{\bn}(\tau, \bx)^{-1} )^{7/4 - l} |\bx|^{-4 + l} \, .
\tag{A.15}
\ee
The factor $c(\tau,r)$ depends on the product of
the constants on the right hand side
of inequalities \eqref{e.a.9} and \eqref{e.a.14}. 
To prove that this bound is integrable,  
spherical coordinates are introduced again.
Since $7/4 - l < 1$, the spherical integral of
\be  \label{e.a.16}
\bex_x \mapsto  ( |\bx| \eta_{\bn}(\tau, |\bx| \bex_x )^{-1} )^{7/4 - l}
= (\sqrt{1 + \tau^2/\bx^2} - (\bn \bex_x))^{-7/4 + l}
\leq (1 - (\bn \bex_x))^{-7/4 + l}
\tag{A.16}
\ee
is bounded. The radial integral of $\bx \mapsto |\bx|^{-4 + l}$
for $|\bx| \geq 2r$ also exists since $4 - l > 3$. 
Since the modulus of these integrals is  
bounded on $\tau_1 \leq \tau \leq \tau_2$,
the remaining $\tau$-integral is finite as well.
Bearing in mind the factor $|\bp|^2$ in equation \eqref{e.a.11},
it follows from \eqref{e.a.13} that
for any $3/4 < l < 1$ and all $\bp$ one has 
\be  \label{e.a.17}
 \Big| \! \int \! d\tau d\bx \, e^{i |\bp| \eta_{\bn}(\tau, \bx) } \, 
 \vartheta_0(\tau) \brho_{\bn}(\bx) \Big| \leq c_l \,
 | \bp |^{l -2} \, . \tag{A.17} 
\ee
The constants $c_l$ depend on the choice of the test functions
$\vartheta_0$, $\vartheta_1$, and $\sigma$ within the above
limitations.
This completes the proof of the bound for $\bm_\infty^\perp$, 
stated in relation~\eqref{e.3.15}, hence of
Proposition \ref{p.3.1}.

\section{Appendix} \label{sectB} 
\setcounter{equation}{0}

We consider here the restriction of the limit
distribution $m_\infty$ analyzed in Appendix \ref{sectA} 
to the region $V_+ \backslash (V_- + t) $, where $t$ is a
sufficiently large time translation
such that the support of $m_1$ is contained in
$V_- + t$. As we will show, this restriction 
coincides with the corresponding restriction of a 
more regular function~$m_{reg}$ that has its support
in $V_+ \cup (V_- + t)$
and whose transverse part is an element
of $L_0(\RR^4)$. By the explanations in Sect.\ \ref{sec3},
this completes the proof of Proposition~\ref{p.3.2}. 

\medskip 
As outlined in Sect.\ \ref{sec3}, we proceed from the 
paths \eqref{e.3.3} on the time shells $\tau \in \mbox{supp} \, \vartheta_0$, 
\be \label{e.b.1}
u \mapsto w_{\tau,\bsx,\bsy}(u) \coloneqq \big( \sqrt{\tau^2 +
  (\bsx + u \bsy)^2} \, ,
\bsx + u \bsy \big) \, ,
\ee
where $\bsx \in \mbox{supp} \, \vartheta_1$, 
$\bsy \in \mbox{supp} \, \sigma$ and $u \geq 0$.
Plugging these data into equation \eqref{e.3.8} and extending
the $u$-integration to infinity, the function $m_\infty$ obtains. 
For large $u$ these paths behave like
\be \label{e.b.2}
 u \mapsto w_{\tau,\bsx,\bsy}(u) = u \, l + r + u^{-1} a + 0(u^{-2}) \, .
\ee
Here the vector 
$l \coloneqq (|\bsy|, \bsy)$ is light-like,
$r \coloneqq \big((\bsx \bsy)/|\bsy|, \bsx \big)$ is
space-like and orthogonal to~$l$, 
and the acceleration $a \coloneqq (1/2)\big((\tau^2 + \bsx^2)/|\bsy| -
(\bsx \bsy)^2/|\bsy|^3, \, \pmb{0} \big) $ is positive timelike.
The bound $0(u^{-2})$ on the remainder
is uniform for $\tau, \bsx, \bsy$
in the supports of $\vartheta_0, \vartheta_1$, and~$\sigma$,
respectively.
Note that the division by $|\by|$ does not cause
problems since the support of~$\sigma$ does not
contain $0$. We also note that 
the $u$-derivative of $w_{\tau,\bsx,\bsy}$ has the 
asymptotic behavior resulting by differentiation of the
terms in equation \eqref{e.b.2}. 

\medskip 
We proceed from these paths to corresponding counterparts in $V_-$,
putting for $u \leq 0$   
\be 
\label{e.b.3}
u \mapsto \underline{w}_{\tau,\bsx,\bsy}(u) \coloneqq \big( - \sqrt{\tau^2 +
  (\bsx + u \bsy)^2} \, , \bsx + u \bsy \big) \, .
\ee
Their leading asymptotic behavior for large negative $u$ 
coincides with the right hand
side of equation~\eqref{e.b.2}. Thus the deviation of
both paths from a straight light-like line is due to the
acceleration~$a$, which causes radiation. Plugging
$\underline{w}_{\tau,\bsx,\bsy}$ into equation~\eqref{e.2.8}, where the
$u$-integral now extends over the negative reals,
one obtains a function $\underline{m}_\infty$
that fits ${m}_\infty$ and has support in
$V_-$.
Since it is obtained from ${m}_\infty$ merely by a  
reflection of time and the direction of paths,
its Fourier transform has the same properties as those established
for ${m}_\infty$ in Sect.~\ref{sec3}. 

\medskip 
To merge these two functions, we modify the initial and final 
parts of the underlying paths. We restrict the $u$-integral
in ${m}_\infty$ to $u \geq 1$ and that in 
$\underline{m}_\infty$ to $u \leq -1$.
The resulting truncated
functions ${m}_\infty^t$, $\underline{m}_\infty^t$ 
differ from the corresponding original functions by test functions,
cf.\ the remark after equation \eqref{e.3.9}.
To fill the gap between the truncated paths, we connect  
each final point $\underline{w}_{\tau,\bsx,\bsy}(-1)$
with the corresponding initial point $w_{\tau,\bsx,\bsy}(1)$ by a straight
line,
\be \label{e.b.4}
u \mapsto (1/2)\big( (w_{\tau,\bsx,\bsy}(1) + \underline{w}_{\tau,\bsx,\bsy}(-1))
+ u \, (w_{\tau,\bsx,\bsy}(1) - \underline{w}_{\tau,\bsx,\bsy}(-1)) \big) \, ,
\quad -1 \leq u \leq 1 \, .
\ee
Plugging these paths into equation \eqref{e.3.8} with the integral
extending over $-1 \leq u \leq 1$, one obtains a  function
${m}^i$ which interpolates between
$\underline{m}_\infty^t$ and ${m}_\infty^t$.
Its Fourier transform
\begin{align} \label{e.b.5}
& p \mapsto \widetilde{m}^i(p) \nonumber \\ 
& = \int \! d\tau d\bsx d\bsy \, \vartheta_0(\tau)
\vartheta_1(\bsx) \sigma(\by)
\big(e^{ipw_{\tau, \bsx, \bsy}(1)} - e^{ip\underline{w}_{\tau, \bsx, \bsy}(-1)}\big) /
p(w_{\tau, \bsx, \bsy}(1) - \underline{w}_{\tau, \bsx, \bsy}(-1)) 
\end{align}
is regular and bounded. The vector
$(w_{\tau, \bsx, \bsy}(1) - \underline{w}_{\tau, \bsx, \bsy}(-1))$
in the denominator 
is positive timelike and varies smoothly with regard to
$\bsx, \bsy$, and $\tau > 0$.
If one differentiates the exponential functions with respect to
$\tau$ and takes into account the support properties of
the test functions, one finds by
partial integration that the
restriction of $\widetilde{m}^i$ 
to positive light-like vectors
converges rapidly to $0$ for large~$p_0 = |\bsp|$.
  
\medskip 
The function $m_{reg}$ that extends 
${m}_{\infty}$ as described above, is now defined by
\be \label{e.b.6}
m_{reg} \coloneqq \underline{m}_{\infty}^t + m^i + {m}_{\infty}^t \, . 
\ee
The difference 
$(m_\infty - m_{reg})$  has support in
$V_- + t$ by construction. So the action
of the functional $\varphi_{m_\infty}$ 
on test functions with
support in $(V_+ + t)$ coincides with that of 
$\varphi_{m_{reg}}$. By the arguments
given in the last paragraph  of Sect.\ \ref{sec3}, this 
applies also to the transverse parts of their spatial
components, \ie to the radiation field. 

\medskip
For the proof that the transverse part of $m_{reg}$ 
is an element of $L_0(\RR^4)$, we need to show
that the restriction of its Fourier transform
to positive light-like momenta is square integrable
with regard to the measure $d \bsp / | \bsp |$. 
For large $| \bsp |$ this follows for both 
$ \widetilde{m}_\infty^t$ and~$\widetilde{\underline{m}}_\infty^t$ 
from the estimate \eqref{e.3.15},
bearing in mind that the differences between the
original and the truncated
functions are test functions.
The intermediate function $\widetilde{m}^i$
is square integrable according to the preceding
observations. So it remains to show that $\widetilde{m}_{reg}$ is 
sufficiently well-behaved for small momenta.

\medskip
Before going into details, we briefly outline the necessary steps.
We consider for given $\tau, \bsx, \bsy$ the continuous
path $u \mapsto z_{\tau, \bsx, \bsy}(u)$ consisting of 
the three pieces fixed by
equations~\eqref{e.b.3} for $u \leq -1$,
\eqref{e.b.4} for $-1 \leq u \leq 1$, and
\eqref{e.b.1} for $u \geq 1$. This path is
smooth apart from the two contact points,
where its derivative is discontinuous,
though bounded. Plugging this path into
equation \eqref{e.3.4}, where the integral now 
extends over the entire real line, and proceeding to the 
Fourier transform, one obtains the distribution 
\be \label{e.b.7}
p \mapsto \widetilde{m}_{z_{\tau, \bsx, \bsy}}(p) \coloneqq 
(2 \pi)^{-2} \int_{-\infty}^\infty \! du \, \dot{z}_{\tau, \bsx, \bsy}(u) \, 
e^{i p z_{\tau, \bsx, \bsy}(u)} \, . 
\ee
Integrating it with the test functions
$\vartheta_0, \vartheta_1$, and $\sigma$, one
arrives at $\widetilde{m}_{reg}$. 
Its behavior for small momenta is determined by the
singularities of $\widetilde{m}_{z_{\tau, \bsx, \bsy}}$. 
The leading singularity is due to 
the asymptotically light-like parts of  $z_{\tau, \bsx, \bsy}$, \ie 
$u \mapsto l_{\tau, \bsx, \bsy}(u) \coloneqq  u l + r$,
cf.\ equation~\eqref{e.b.2}. 
As we will see, the transverse part of the resulting
distribution vanishes. So this distribution
does not contribute to the radiation field $\widetilde{m}_{reg}^\perp$,
in accordance with the fact that charges with constant velocity
do not radiate. We may therefore subtract it
from equation \eqref{e.b.7}. The resulting difference is governed by the 
acceleration $a$ in  equation~\eqref{e.b.2} that 
gives rise to a moderate logarithmic
singularity if $p$ becomes parallel to $l_{\tau, \bsx, \bsy}$.
It  implies that the restriction of the transverse function 
$\widetilde{m}_{z_{\tau, \bx, \by}}^\perp$ to positive lightlike
momenta is square integrable with regard to $d \bsp / | \bsp |$
in a neighborhood of the origin. 

\medskip
Turning to the details, we replace 
in equation \eqref{e.b.7} the path $z_{\tau, \bsx, \bsy}$
with~$l_{\tau, \bsx, \bsy}$. It yields the distribution
$p \mapsto (2 \pi)^{-1} \delta(l p) e^{ipr} l$. 
Its restriction to light-like momenta
vanishes unless~$p$ is parallel to $l$. 
Hence the projection onto its 
transverse component vanishes and we can therefore subtract
this term from \eqref{e.b.7}. 
To determine the effect of the acceleration~$a$ in
\eqref{e.b.2}, we consider the truncated 
paths $u \mapsto a_{\tau, \bsx, \bsy} \coloneqq u l + r + u^{-1} a$
for $\pm u \geq 1$. The difference between the
corresponding distribution and the one 
for the light-like path is, for~$u \geq 1$, 
\be \label{e.b.8}
p \mapsto \int_1^\infty \! du \, \big( \dot{a}_{\tau, \bsx, \bsy}(u)
e^{i p {a}_{\tau, \bsx, \bsy}(u)} 
- \dot{l}_{\tau, \bsx, \bsy}(u) e^{i p {l}_{\tau, \bsx, \bsy}(u)} \big) \, ,
\ee
and a similar integral is obtained for $u \leq -1$.

\medskip 
The modulus of the  integral \eqref{e.b.8} is bounded by
\begin{align} \label{e.b.9}
& \Big| \int_1^\infty \! du \,  
  ( (l  - u^{-2} a) \, e^{ip (u l + u^{-1}a)} - l  e^{iu p l} \big)
  \Big| \nonumber \\
&  \leq |a| \int_1^\infty \! \! du \, u^{-2} +
  | l | \, \Big| \! \int_1^\infty \! du \, e^{iu pl} (1 - e^{i u^{-1} p a} )
  \Big| \nonumber \\
& \leq \big(|a| + 4 | l | |p a|^2 \big) \int_1^\infty \! \! du \, u^{-2} + 
  | l | \, | pa | \, \Big| \!
  \int_1^\infty \! du \, u^{-1} e^{iu pl} \Big| \nonumber \\
& = \big(|a| + 4 | l | |p a|^2 \big)
  +  | l | \, | pa | \, \Big| \! 
  \int_{p l}^\infty \! du \, u^{-1} e^{iu} \Big| \, .
\end{align}
The integral in the last line is finite if $pl > 0$
and diverges logarithmically
if $pl$ approaches~$0$. A similar estimate holds for the
integral over the negative reals. 
Thus the restriction of the function
\eqref{e.b.8} to positive light-like momenta
is square integrable  with regard to $d\bsp/| \bsp |$
for small $\bsp$, and the same holds for the 
integral over the negative reals. 

\medskip
The distribution determined by the original path 
$u \mapsto z_{\tau, \bsx, \bsy}(u)$ has the same
singularities for small momenta as the one determined by
$u \mapsto a_{\tau, \bsx, \bsy}(u)$.
To verify this, one replaces in equation \eqref{e.b.8}
the path $a_{\tau, \bsx, \bsy}$ by $z_{\tau, \bsx, \bsy}$
and $l_{\tau, \bsx, \bsy}$ by $a_{\tau, \bsx, \bsy}$.
The modulus of the resulting difference can be estimated by 
\begin{align} \label{e.b.10} 
& \big| \int_{1}^\infty \! du \,
  \big( \dot{z}_{\tau,\bsx,\bsy}(u)
  e^{ipz_{\tau, \bsx, \bsy}(u)} - \dot{a}_{\tau,\bsx,\bsy}(u)
  e^{ipa_{\tau, \bsx, \bsy}(u)} \big| \nonumber \\ 
  & \leq \int_{1}^\infty \! du \,
  | \dot{z}_{\tau, \bsx, \bsy}(u) -  \dot{a}_{\tau, \bsx, \bsy}(u)|
+ (| l | + | a |)| p | \int_{1}^\infty \! du \,
|z_{\tau, \bsx, \bsy}(u) -  a_{\tau, \bsx, \bsy}(u) | \, .
\end{align}
Now  according to relation \eqref{e.b.2} one has  
$u \mapsto \big( z_{\tau, \bsx, \bsy}(u) -  a_{\tau, \bsx, \bsy}(u) \big) = 0(u^{-2})$
for large~$u$.  
The derivative with regard to~$u$ behaves like
$ 0(u^{-3})$. So the integral is bounded for small~$p$ and this  
holds for its transverse part as well.
The same result obtains for the integral over the
negative reals, $u \leq - 1$.

\medskip
The difference  between the distribution determined by the intermediate part 
of $z_{\tau, \bsx, \bsy}$, cf.\ equation \eqref{e.b.4},
and the lightlike path can be estimated by  
\begin{align} \label{e.b.11} 
p & \mapsto \Big|  \int_{-1}^1 \! du \,
  \big( \dot{z}_{\tau,\bsx,\bsy}(u)
  e^{ipw_{\tau, \bsx, \bsy}(u)} - \dot{l}_{\tau,\bsx,\bsy}(u)
  e^{ipl_{\tau, \bsx, \bsy}(u)} \big)  \Big| \nonumber \\
  & \leq 2 \, ( |w_{\tau, \bsx, \bsy}(1) -
  \underline{w}_{\tau, \bsx, \bsy}(-1)| + | l | )
  \leq 2(\sqrt{\tau^2 + \bx^2 + 2 \by^2} + | l|) 
  \, .
\end{align}

\medskip
The preceding three estimates hold for 
$\tau, \bx$, and $\by$ in the
compact support of the test functions. The results
show that the difference
$p \mapsto (\widetilde{m}_{z_{\tau, \bx, \by}} - \widetilde{m}_{l_{\tau, \bx, \by}})(p)$,
restricted to positive light-like momenta, is square integrable for
small momenta with regard to the measure~$d\bp/|\bp|$.
A straightforward examination of the
corresponding integrals also shows that
they are uniformly bounded with respect to $\tau, \bx, \by$ varying
in the supports of the given test functions. 
So~$\widetilde{m}_{reg}^\perp$ is square integrable
with regard to the measure in a neighborhood of the origin
as well. Summarizing these observations, it follows
that~$\widetilde{m}_{reg}^\perp \in L_0(\RR^4)$, completing the
proof of Proposition \ref{p.3.2}. 

\section{Appendix} \label{sectC} 
\setcounter{equation}{0}

In this appendix we provide the proof of Lemma \ref{l.4.1},   
which states that every state $\omega_0 \beta_{\bl}$ on $\, \pmb \efV$ 
with $\bl \in L_1(\RR^4)$ induces
a representation in which the spacetime translations are
unitarily implemented so that the relativistic
spectrum condition holds. A similar result was obtained by 
Roepstorff \cite{Ro} who considered the
transverse  radiation fields. Since
we are also interested in the gauge bridges, we need to consider
in addition the longitudinal components. Instead of adopting the
strategy of Roepstorff for the proof, we outline here a slightly different
argument that reveals the ingredients entering in this result. 

\medskip
We begin by noting that the GNS-representation of the
algebra $\, \pmb{\efV}$ induced by
$\omega_0 \beta_{\bl}$  
can be realized in the vacuum representation induced by $\omega_0$.
There it is given by the action of $\beta_{\bl}$ on the algebra $\, \pmb{\efV}$
in the vacuum Hilbert space. If $\bl$ is an element of the
subspace $L_0(\RR^4) \subset L_1(\RR^4)$, the corresponding 
$\beta_{\bl}$ is given by the adjoint action of the
unitary operators $e^{i \bA(\bl)}$. It follows that 
the unitary spacetime translations in this representation 
are given by, $x \in \RR^4$, 
\begin{align} \label{e.c.1}
U_{\bl}(x) & = e^{i \bA(\bl)} U_0(x) e^{- i \bA(\bl)} \nonumber \\
&  = e^{i \bA(\bl)}  e^{- i \bA(\bl(x))} \, U_0(x) =
e^{-(i/2)\langle \bl, D \, \bl(x) \rangle } e^{i \bA(\bl - \bl(x))} \, U_0(x) \, .
\end{align}
Here $U_0$ is the representation of the translations in the vacuum
representation and $\bl(x)$ denotes the function $\bl$, translated by $x$.
In the third equality we made use of the Weyl relations \eqref{e.2.2}. 
The first equality shows that $U_{\bl}$ is a continuous unitary
representation of the translations which satisfies the relativistic
spectrum condition. Turning to the last term, note that
\be
-(i/2)\langle \bl, D \, \bl(x) \rangle
= (i/2)\langle \bl, D \, (\bl - \bl(x)) \rangle
= i \mbox{Im} \, \langle \bl, (\bl - \bl(x)) \rangle_0 \, ,
\ee
where $\langle \cdot  \hspace{2pt} , \cdot \rangle_0$ is the scalar product
in~$L_0(\RR^4)$. Since small momenta in $\bl$ are suppressed in the 
differences $(\bl - \bl(x))$ appearing in the last term of \eqref{e.c.1},  
one can extend
this relation to all $\bl \in L_1(\RR^4)$, while retaining the
properties of $U_{\bl}$. We briefly sketch the argument.

\medskip
Given any $\bl \in L_1(\RR^4)$, we approximate it by elements
$\bl_\varepsilon \in L_0(\RR^4)$, $\varepsilon > 0$.
Choosing a smooth function $p \mapsto \chi_\varepsilon(p)$
which is equal to $0$ for $|p| \leq \varepsilon$ and equal to
$1$ for $| p | \geq 2 \varepsilon$, these approximants are defined
in momentum space by 
$p \mapsto \widetilde{\bl_\varepsilon}(p) \coloneqq \chi_\varepsilon(p)
\widetilde{\bl}(p)$. So $\bl_\varepsilon$ converges
to $\bl$ in $L_1(\RR^4)$ in the limit of small $\varepsilon$, hence
$\beta_{\bl_\varepsilon}$ converges pointwise in norm to~$\beta_{\bl}$
on the Weyl operators spanning $\, \pmb{\efV}$. The essential point
is that $(\bl_\varepsilon - \bl_\varepsilon(x))$ converges
strongly in $L_0(\RR^4)$, uniformly
for $x$ in compact sets $\mC$.  This is apparent in momentum space and we 
omit the proof. It follows that the unitary Weyl operators
$e^{i A(\bl_\varepsilon -  \bl_\varepsilon(x))}$ converge 
for small $\varepsilon$ in the strong operator topology,
uniformly on $\mC$. Similarly, one
finds that the phases  
$\langle \bl_\varepsilon , (\bl_\varepsilon - \bl_\varepsilon(x) \rangle $
converge.
So the operators $U_{\bl_\varepsilon}(x)$
converge in the limit of small $\varepsilon$ in the strong operator
topology, uniformly for $x \in \mC$. Their limits
define unitary operators $U_{\bl}(x)$, $x \in \RR^4$. They form
a continuous unitary representation of the spacetime translations
which satisfies the relativistic spectrum condition. It also 
follows from the preceding discussion
that their adjoint action generates the translations in the
representation $\beta_{\bl}$. 
This completes the proof of Lemma \ref{l.4.1}. 

\end{appendix}

\section*{Acknowledgment}
DB  gratefully acknowledges the 
support of Roberto Longo and the 
University of Rome ``Tor Vergata'', which made
this collaboration possible. He is also grateful
to Dorothea Bahns and the Mathematics Institute
of the University of G\"ottingen for their continuing
hospitality. EV, FC and GR are grateful to Karl-Henning Rehren
for his kind hospitality and support at the University of
G{\"o}ttingen. FC and GR also acknowledge GNAMPA-INdAM.
FC acknowledges the PRIN PNRR 2022 CUP E53D2301806 0001.
GR acknowledges the MIUR Excellence Department Project awarded
to the Department of Mathematics, University of Rome Tor Vergata,
CUP E83C23000330006.

\end{document}